\documentclass[conference]{IEEEtran}

\usepackage[utf8]{inputenc}
\usepackage[english]{babel}
\usepackage{hyperref}
\usepackage{amssymb}
\usepackage{amsmath}
\usepackage{graphicx}
\usepackage{amsthm}
\usepackage{bbm}
\usepackage{mathtools}
\usepackage{setspace}
\usepackage{subfigure}
\usepackage{siunitx}
\usepackage{url}
\usepackage{tikz}

\newtheorem{lemma}{Lemma}

\newcommand\copyrighttext{%
\footnotesize \textcopyright \enspace 2015 IEEE. Personal use of this material is permitted. Permission from IEEE must be obtained for all other uses, in any current or future media, including reprinting/republishing this material for advertising or promotional purposes, creating new collective works, for resale or redistribution to servers or lists, or reuse of any copyrighted component of this work in other works.  DOI: \href{https://doi.org/10.1109/ICCW.2015.7247332}{10.1109/ICCW.2015.7247332}
}
\newcommand\copyrightnotice{%
\begin{tikzpicture}[remember picture,overlay]
\node[anchor=south] at (current page.south) {\fbox{\parbox{\dimexpr\textwidth-\fboxsep-\fboxrule\relax}{\copyrighttext}}};
\end{tikzpicture}%
}

\begin{document}

\title{Modelling Machine Type Communication \\ in IEEE 802.11ah Networks}

\author{\IEEEauthorblockN{Evgeny Khorov$^{+,*}$, Alexander Krotov$^{*}$, Andrey Lyakhov$^{*}$}
\IEEEauthorblockA{$^{*}$Institute for Information Transmission Problems,  Russian Academy of Sciences, Moscow, Russia\\
$^{+}$Moscow Institute of Physics and Technology, Moscow, Russia 
\\  Email: \{khorov, krotov, lyakhov\}@iitp.ru}
}

\maketitle
\copyrightnotice
\begin{abstract}
Wi-Fi was originally designed to provide broadband wireless Internet access for devices which generate rather heavy streams. And Wi-Fi succeeded. The coming revolution of the Internet of Things with myriads of autonomous devices and machine type communications (MTC) traffic raises a question: can the Wi-Fi success story be repeated in the area of MTC? Started in 2010, IEEE 802.11 Task Group ah (TGah) has developed a draft amendment to the IEEE 802.11 standard, adapting Wi-Fi for MTC requirements. The performance of novel channel access enhancements in MTC scenarios can hardly be studied with models from Bianchi's clan, which typically assume that traffic load does not change with time. This paper contributes with a pioneer analytical approach to study Wi-Fi-based MTC, which can be used to investigate and customize many mechanisms developed by TGah. \footnote{The research was done in IITP RAS and supported by the Russian Science Foundation (agreement No 14-50-00150)}

\end{abstract}

\section{Introduction}

Wi-Fi was originally designed to provide broadband wireless Internet access for personal computers and laptops. In 2000s it connected  smartphones and other gadgets generating heavy streams to the Internet. Wi-Fi succeeded: in a modern city, one can hardly find a Wi-Fi-less place.

Continuously increasing user demands for more and more throughput is reflected in the IEEE 802~LAN/MAN Standards Committee (LMSC) and Wi-Fi Alliance activities, which result in dozens of IEEE 802.11 standard amendments improving Wi-Fi. In addition to high speed anytime and anywhere Internet access, the concept of smart city includes processing data gathered from myriads of autonomous battery-powered devices: various sensors, cameras, etc. Thus, the recent challenges are connected with machine type communications (MTC), since the MTC market with its tremendous number of devices is very attractive to industry. A typical MTC scenario is providing Internet access for thousands of sensors and actuators (e.g. parking sensors, water meters, fire detectors, etc.) communicating with rarely sent very short messages. Although some other technologies --- ZigBee, RFID, etc. --- are already being used for MTC, they have their own limitations, and the Wi-Fi community sees an opportunity for outperforming them and extending Wi-Fi application area to MTC scenarios. To address MTC issues, IEEE~802~LMSC is currently developing the .11ah amendment \cite{802.11ah} for the IEEE~802.11 standard, which seems to be ready by 2016\footnote{Another .11ah use case, offloading, is left beyond the scope of the paper.}.

The IEEE~802.11ah network shall support up to 6000 low power stations (STAs) simultaneously connected to an access point (AP), which results in extremely high contention, high number of collisions and transmission retries compared to usual Wi-Fi networks. Every transmission attempt consumes battery power and reduces overall STA lifetime. The situation ever worsens if more powerful offloading STAs are present in the neighborhood. Thus, power efficient channel access for a crowd of battery powered STAs is a challenging task.  

The straightforward approach for this task is to allocate time intervals and to assign each interval to a single STA, as it is done in MCCA --- deterministic channel access in Wi-Fi Mesh networks \cite{automation}. Although increasing the throughput in saturated networks, such a collision free approach results in extremely long delays in normal load conditions and can hardly be used for communication of thousands of STAs with sporadic and unpredictable traffic triggered by external events.
To reduce the number of retries, and thus to extend battery lifetime, new contention limiting mechanisms have been included into the standard draft \cite{802.11ah}. One of them is Restricted Access Window (RAW), a novel hierarchic channel access mechanism, which assigned some time intervals for a group of STAs and then they contend for the channel in the assigned intervals. The core contribution of the paper is a pioneer mathematical approach which can be used to evaluate RAW efficiency and to adaptively choose its parameters in typical MTC scenarios. Apart from that, the approach can be used to study and customize parameters of some other Wi-Fi mechanisms, which are based on similar hierarchic channel access mechanisms.

The rest of the paper is organized as follows. In Section~\ref{sec:RAW} we give a simplified description of RAW. For the detailed one, please refer to \cite{comcom11ah}. Section~\ref{sec:problem} describes a scenario typical for MTC communication and proposes a \emph{general} problem statement. In Section~\ref{sec:related}, we explain why the numerous existing Wi-Fi models cannot be used to solve this problem. We develop our model in Section~\ref{sec:model} and present numerical results in Section~\ref{sec:results}. Finally, Section \ref{sec:concl} concludes the paper.

\section{Restricted Access Window}

\label{sec:RAW}

The main idea of RAW is to reduce the number of STAs concurrently accessing the medium and to distribute channel accesses over time. For that, the AP selects a group of STAs and assigns it to a time interval, called the RAW slot. The STAs are forbidden to transmit in alien RAW slots.

By dividing STAs into groups and assigning them RAW slots, the AP reduces contention. However, the AP typically does not know in advance which STAs have frames for transmission. In this case, assigning RAW slots individually to each STA results in excess consumption of channel resources and reduces throughput. For this reason, a group may contain a \emph{large} number of STAs, while only some of them have frames for transmission in the beginning of the RAW slot. In particular, the standard assumes that the AP may take into account STA type, power constraints, traffic pattern while grouping STAs. Having estimated the number of STAs having data to transmit, the AP selects the RAW slot duration and position.

The AP periodically broadcasts all RAW parameters in beacons, letting STAs know which group they belong to and when the group RAW slot occurs. Since transmissions outside the RAW slot are not protected from collisions at all, it is reasonable for STAs to sleep always except for their RAW slots and beacon transmission times, saving energy.

Inside their RAW slots, STAs transmit with the legacy EDCA (Enhanced Distributed  Channel Access). EDCA implements the CSMA/CA (Carrier Sense Multiple Access with Collision Avoidance) method. It means that before starting transmission the STA senses medium and does not transmit until the medium is idle. Besides that, to reduce the probability of a collision right after the medium becomes idle, EDCA uses truncated binary exponential backoff.

Since the contention conditions inside and outside the RAW slot differ, the STA uses two different backoff functions inside and outside the RAW slot. In particular, when its RAW slot begins, the STA creates a new backoff function and initializes backoff counter with a random integer value drawn from the uniform distribution over interval $[0,CW_0 - 1]$. Then the STA starts sensing the medium. Each time, the medium is idle for backoff slot $\sigma$, the STA decrements the backoff counter. If the medium becomes busy, the STA freezes the backoff counter.  The backoff counter is resumed and decremented after the medium is idle for some time. This time equals $AIFS$ if the STA received a frame successfully, or $EIFS=T_{ACK} + AIFS$ if the STA was not able to successfully decode frame, where $T_{ACK}$ is the time needed to transmit an acknowledgment frame (ACK).

When the backoff counter reaches 0, the STA checks if it can transmit its frame and receive the ACK within the RAW slot. If the frame exchange sequence crosses the RAW slot boundary, the STA does not transmit and can switch to the doze state. Otherwise, it transmits the frame and waits for an ACK from the AP. If the ACK is received within $T_{ACK}$, the STA considers that the frame has been successfully transmitted.  If no ACK is received and retry limit $RL$ is not reached, the STA starts the next transmission attempt. Prior transmission attempt $i$ of a packet, the STA draws new backoff counter value from the uniform distribution over the interval $[0,CW_{i - 1} - 1]$, where
\[
CW_i = \begin{cases}
CW_{min}, & i = 0,\\
\min \{CW_{max}, 2 CW_{i - 1}\}, & i > 0.
\end{cases}
\]

\section{Problem Statement}
\label{sec:problem}

Consider a network with many power limited STAs communicating with an AP. Let $N$ STAs start accessing the channel with EDCA at the beginning of the RAW slot. Since MTC traffic consists of rarely sent messages, let each of the $N$ STAs have the only frame for transmission. Having transmitted its frame, STA stops contending for the channel and can return to the doze state. Thus, the number of STAs accessing the channel decreases with time.

Let the STAs be located within the transmission range of each other and transmission errors be only caused by collisions. We also assume that the AP does not transmit anything except for the ACKs.

Let us find the minimal RAW slot duration required for:

\begin{enumerate}
\item[A.] an arbitrary chosen STA to successfully transmit its frame with some predefined probability;
\item[B.] all $N$ STAs to successfully transmit their frames with some predefined probability.
\end{enumerate}

To solve these problems, in Section~\ref{sec:model} we find:

\begin{enumerate}
\item[A.] the distribution of time $P_A^N(\tau)$ needed for an arbitrary STA to successfully transmit its frame;
\item[B.] the distribution of time $P_B^N(\tau)$ needed for all $N$ STAs to transmit their frames\footnote{Note that problems A and B are not reducible to each other.}.
\end{enumerate}

Such a model can be applied for various purposes. First of all, it can be used to find the distribution of energy consumption during the accessing in RAW process, which is rather important for energy harvesting devices powered by a small accumulator or even capacitor. Apart from that, the model can be used to study other mechanisms, e.g. Wi-Fi Power Management framework which works as follows. The AP buffers frames which are destined for STAs in the power saving (PS) mode. From time to time, these STAs wake up to receive beacons. Beacons contain Traffic Indication Map (TIM) information elements indicating for which STAs the AP has buffered frames. If the AP has no frames destined for a STA, the STA can switch to the doze state. Otherwise, it can retrieve these frames by sending special PS-Poll frame by means of EDCA. As the response to these frames, the AP send buffered data. To protect PS-Poll from collisions, the AP may use RAW or other protection mechanisms, restricting non-PS STAs accessing the channel during some time interval. The model allows to choose the interval duration.

\section{Related Work}
\label{sec:related}

The widest known mathematical model of DCF --- the basic random channel access used in Wi-Fi networks --- was developed by Bianchi in \cite{bianchi2000performance}. The model allows estimating maximal throughput, assuming that a constant number of active STAs work in saturated conditions. So the model cannot be used to solve the problems stated in Section \ref{sec:problem}, since in these problems the number of active STAs decreases with time. However, paper \cite{bianchi2000performance} contains basic principles of Wi-Fi modeling. In particular, it introduces a concept of virtual slot, which is the time interval between consequent backoff counter changes.

Paper \cite{zheng2013performance} presents a model, which allows estimating the maximal throughput (again, in saturated scenarios)  if all STAs are equally divided into several groups and each slot is assigned to a group. It proves that RAW manifold increases throughput in a network with thousands STAs, however the model can not be used for our problems for aforesaid reasons.

In \cite{buratti2010be}, the authors consider another protocol, IEEE 802.15.4 that uses similar to EDCA channel access. However, in .15.4 a STA senses the channel only when the backoff ends. Although both papers present performance evaluation in the scenario similar to the one described in this paper, the authors assume that collision probability is constant, while in reality, both varying contention window and the number of STAs having packets to transmit make collision probability change with time.

The authors of \cite{liu2013power} study the power saving  mechanism.
They have developed a model, which allows estimating the average energy consumed by a STA and average time used by a STA to retrieve its data. As shown in Section~\ref{sec:results}, even though the model developed in \cite{liu2013power} can be used to find the average frame transmission time for a STA, it can not be used to find the correct time distribution required in problem A. Besides that, it can not be used at all to solve problem B.

So, we can conclude that a new mathematical model is required to solve the problems A and B.

\section{Model}
\label{sec:model}

\subsection{Markov Processes}

The core contribution of the paper is the model of the  described in Section \ref{sec:problem} process of the channel access of $N$ STAs, each of which has a frame in the beginning of some limited time interval, i.e. the RAW slot. To simplify further description, we assume that all frames are of the same size, however the model can be easily extended to the general case.

The model consists of two Markov chains referred to as process A and process B. These processes describe the behaviour of an arbitrarily selected STA and all STAs, respectively, to solve the problems defined in Section~\ref{sec:problem}. Both processes use the discrete and integer time scale with a timeunit equal to a virtual time slot. Note that this discrete time scale does not directly relate to the system time. Depending on the state of the virtual slot -- empty, containing a successful transmission or a collision -- its duration is $T_{e}$, $T_{s}$ or $T_{c}$, respectively.  Since all STAs are in the transmission range of each other, any slot is of the same type for every STA.

We denote the state of process A as $(t, c, s, r)$, where $t$ is the model time, i.e. the number of virtual slots since the beginning of the RAW slot ($t \ge 0$), $c$ and $s$ are the numbers of collision and successful slots, $r$ is the retry counter value of the chosen STA, i.e., the number of unsuccessful transmission attempts made by the STA prior to slot $t$. The number of empty slots till slot $t$
is $t - c - s$. When the STA successfully transmits its frame or $r$ reaches the retry limit $RL$, the process goes into the successful or unsuccessful absorbing state, correspondingly.

We denote the states of process B as $(t, c, s)$. The process terminates when all STAs have successfully transmitted their frames, i.e., $s = N$, so states $(t, c, N)$ are absorbing.

The rest of the Section is organized as follows. We state and prove a lemma needed for the analysis of processes A and B in Subsection~\ref{sec:txprobability}. Then in Subsections~\ref{sec:proc-A}~and~\ref{sec:proc-B}, we describe possible transitions and their probabilities in both processes.
In Subsection~\ref{sec:calc}, we show how to compute the state probabilities for both processes.

\subsection{Conditional Probability of a Transmission in a Virtual Slot}
\label{sec:txprobability}

Let us define two events. Event $(t,r)$ means that process A goes through the state with given number $t$ of the whole virtual slots and number $r$ of transmission attempts. Event $TX$ means that the STA transmits. Thus, $\Pr(TX|t, r)$ is the probability of the STA to transmit a packet in slot $t$ provided that by the beginning of this slot it has made $r$ unsuccessful transmissions.

\begin{lemma}
If the number of STAs is infinite ($N \to \infty$), then
\begin{equation}
\label{eq:txprobability}
\Pr(TX|t, r) = \frac{a(t, r)}{b(t, r)},
\end{equation} where
\begin{align*}
a(t,r) &= \begin{cases}
	\frac{1}{CW_0}, & r = 0, 0 \le t < CW_0, \\
	0, & r = 0, t \ge CW_0, \\
	0, & r \ge RL, \\
	\frac{1}{CW_r} \sum \limits _{i = t - CW_r}^{t - 1} a(i,r - 1), & 0 < r < RL;
	\end{cases}\\
b(t,r) &= \begin{cases}
1 - \sum\limits_{i = 0}^{t} a(i, r), & r = 0,\\
\sum\limits_{i = 0}^{t} (a(i, r - 1) - a(i, r)), & r > 0.
\end{cases}
\end{align*}
\end{lemma}

\begin{proof}

It can be shown that
\[
\Pr(TX,t,r) = \begin{cases}
	\frac{1}{CW_0}, & r = 0, 0 \le t < CW_0, \\
	0, & r = 0, t \ge CW_0, \\
	0, & r \ge RL, \\
	\sum\limits_{i = t - CW_r}^{t - 1} \frac{\Pr(C, i, r - 1)}{CW_r}, & 0 < r < RL,
\end{cases}
\]
where $\Pr(C, i, r)$ is the probability of chosen STA to have a collided transmission in slot $i$ with retry counter $r$. The first three lines of the above equation are obvious. The last line corresponds to a transmission retry.   We use it to calculate the probability of retransmission in a given slot: after a collision with $r < RL - 1$ the STA chooses one of the next $CW_r$ slots for retransmission, each with probability $\frac{1}{CW_r}$.

To find $\Pr(t,r)$, let us notice that the retry counter equals $r$ in the beginning of slot $t$ if and only if it became equal to $r$ at the beginning of some previous slot (for $r = 0$, we consider slot 0, thus the probability of such an event is 1, and for $r > 0$, the probability equals $\sum\limits_{i = 0}^{t-1} \Pr(C, i, r - 1)$), and the STA did not transmit since that slot.  Then,
\[
\Pr(t,r) = \begin{cases}
1 - \sum\limits_{i = 0}^{t - 1} \Pr(TX, i, r), & r = 0,\\
\sum\limits_{i = 0}^{t - 1} \Pr(C, i, r - 1) - \sum\limits_{i = 0}^{t - 1} \Pr(TX, i, r), & r > 0.
\end{cases}
\]

For the infinite number of STAs, collision probability $\Pr(C, t, r)$ equals transmission probability $\Pr(TX, t, r)$. So, $\Pr(TX, t, r) = a(t, r)$, $\Pr(t, r) = b(t, r)$,
and $\Pr(TX|t, r) = \frac{a(t, r)}{b(t, r)}$, by definition.
\end{proof}

\begin{figure}[!htbp]
\centering
\includegraphics[width=0.8\linewidth]{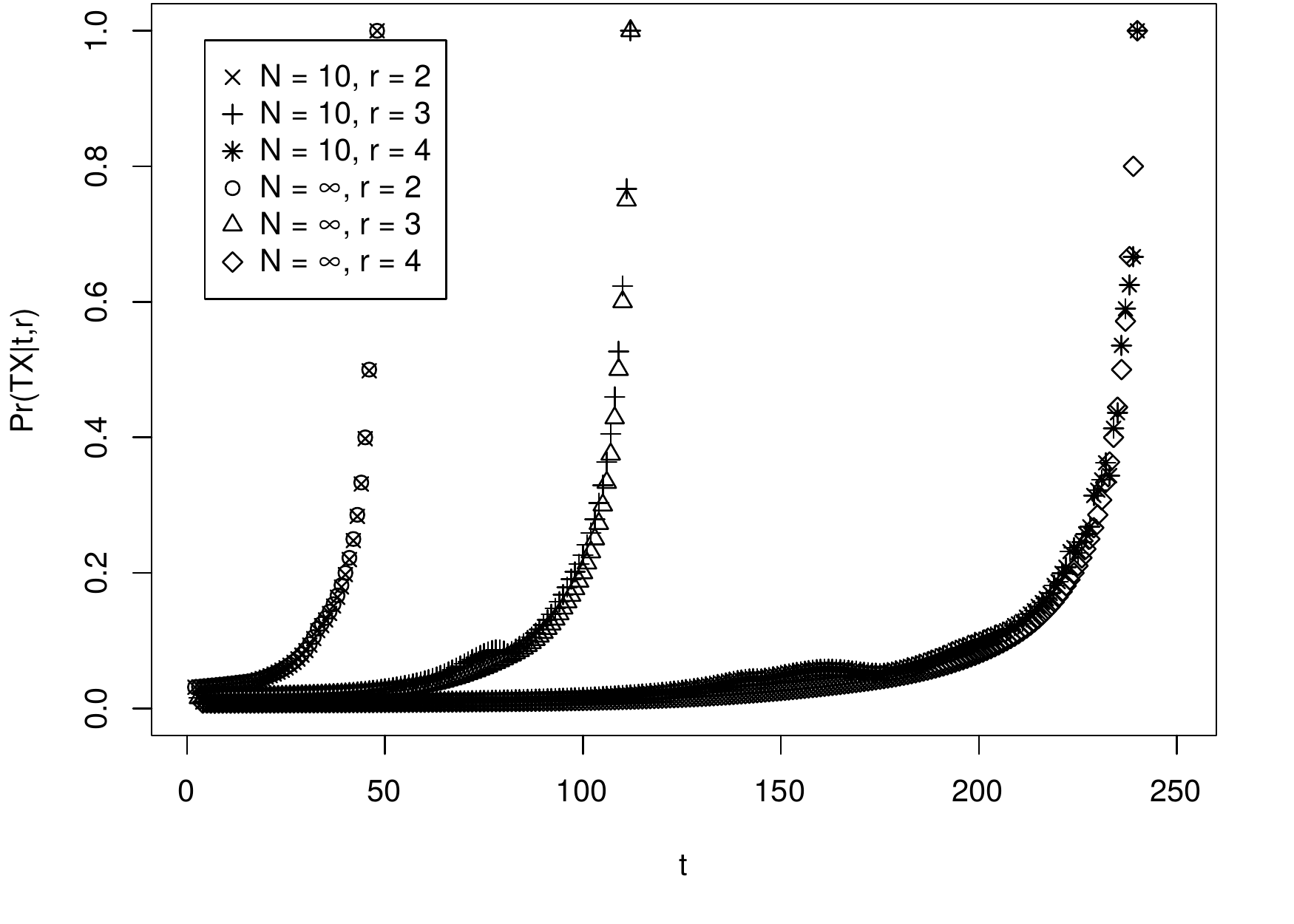}
\caption{\label{fig:tx} Conditional transmission probability $\Pr(TX|t,r)$}
\end{figure}

To simplify further analysis, we assume that $\Pr(TX|t, r)$ does not depend on the number of STAs.  If $r = 0$, it is valid since transmission attempt probability depends only on $CW_{min}$ and does not depend on the number of STAs.  For $r = 2$, $r = 3$ and $r = 4$, the assumption is validated using simulation (see Fig.~\ref{fig:tx}). We can see that the probability for $N = 10$ does not significantly differ from the probability calculated for $N \to \infty$. Moreover, the difference decreases with growth of the $N$. Even when the number of STAs is low, this assumption does not affect dramatically the results obtained for problems A and B, because in this case they likely deliver their frames at the first or at least the second transmission attempt, and when the number of STAs is high, the calculated probability does not significantly differ from the real one even for high values $r$.  For these reasons, we  use~(\ref{eq:txprobability}) for any number of STAs.

\subsection{Process A}
\label{sec:proc-A}

Process A starts in state $(0, 0, 0, 0)$. From state $(t, c, s, r)$ five transitions are possible.
\begin{enumerate}
\item With probability $\Pi_e(t, c, s)$, slot $t$ is empty and the process transits to $(t + 1, c, s, r)$.
\item  With probability $\Pi_s^+(t, c, s)$, the transmission attempt of the chosen STA in slot $t$ is successful and the process transits to the successful absorbing state.
\item With probability $\Pi_s^-(t, c, s)$, the transmission attempt of another STA in slot $t$ is successful and the process transits to $(t + 1, c, s + 1, r)$.
\item With probability $\Pi_c^+(t, c, s)$, slot $t$ is collision and the chosen STA transmits, i.e. the process transits to $(t + 1, c + 1, s, r + 1)$, which can be an unsuccessful absorbing state if $R + 1 = RL$.
\item With probability $\Pi_c^-(t, c, s)$, slot $t$ is collision and the chosen STA does not transmit, so the process transits to $(t + 1, c + 1, s, r)$.
\end{enumerate}
To obtain transition probabilities, let us find probabilities of a slot to be empty, successful or collision, given that the chosen STA does not transmit:
\begin{align*}
	\Pi_e'(t, c, s) & = (1 - \Pr(TX|t, c, s))^{N - s - 1}, \\
	\Pi_s'(t, c, s) & = (N - s - 1) \Pr(TX|t, c, s) \times\\
	                & \times (1 - \Pr(TX|t, c, s))^{N - s - 2}, \\
	\Pi_c'(t, c, s) & = 1 - \Pi_e'(t, c, s) - \Pi_s'(t, c, s),
\end{align*}
where $\Pr(TX|t,c,s)$ is the probability of the chosen  STA to transmit if the process A is in state $(t, c, s)$.

\begin{lemma}
\label{lem:tx}
\[ \Pr(TX|t, c, s) = \frac{\sum_{r = 0}^c \Pr(TX|t, r) \Pr(t, c, s, r)}{\sum_{r = 0}^c \Pr(t, c, s, r)}, \] where $\Pr(t, c, s, r)$ is the probability of process A being in state $(t, c, s, r)$ for given $t$.
\end{lemma}
\begin{proof}
By definition,
\[ \Pr(TX|t, c, s) = \frac{\Pr(TX, t, c, s)}{\Pr(t, c, s)} = \frac{\Pr(TX, t, c, s)}{\sum_r \Pr(t, c, s, r)}. \]
By the law of total probability,
\[ \Pr(TX,t, c, s) = \sum_r \Pr(TX|t, c, s, r) \Pr(t, c, s, r), \]
but $\Pr(TX|t, c, s, r)$ depends on neither $c$ nor the number of STAs accessing the channel, i.e. $N-s$ (according to the assumption from Section \ref{sec:txprobability}). So, $\Pr(TX|t, c, s, r) = \Pr(TX|t, r)$.
Taking into account that $r \leq c$ (the number of transmission attempts made by the chosen STA cannot exceed the total number of collisions), we obtain the Lemma statement.
\end{proof}
Using Lemma~\ref{lem:tx}, we get transition probabilities:
\begin{align*}
\Pi_e^-(t, c, s) &= (1 - \Pr(TX|t, r)) \Pi_e'(t, c, s), \\
\Pi_s^+(t, c, s) &= \Pr(TX|t, r) \Pi_e'(t, c, s), \\
\Pi_s^-(t, c, s) &= (1 - \Pr(TX|t, r)) \Pi_s'(t, c, s), \\
\Pi_c^+(t, c, s) &= \Pr(TX|t, r) (1 - \Pi_e'(t, c, s)), \\
\Pi_c^-(t, c, s) &= (1 - \Pr(TX|t, r)) \Pi_c'(t, c, s).
\end{align*}

\subsection{Process B}
\label{sec:proc-B}

Process B starts in state $(0, 0, 0)$. Similarly, we obtain transition probabilities to states $(t + 1, c, s)$ (empty slot), $(t + 1, c, s + 1)$ (successful), and $(t + 1, c + 1, s)$ (collision):
\begin{align*}
	\Pi_e(t, c, s) &= (1 - \Pr(TX|t, c, s))^{N - s}, \\
	\Pi_s(t, c, s) &= (N - s) \Pr(TX|t, c, s) (1 - \Pr(TX|t, c, s))^{N - s - 1}, \\
	\Pi_c(t, c, s) &= 1 - \Pi_e - \Pi_s.
\end{align*}

\subsection{Computation}
\label{sec:calc}

We obtain probabilities of the processes states step by step, increasing $t$ sequentially and starting with $t = 0$. At each step $t$, we compute:
\begin{enumerate}
\item $\Pr(TX|t, c, s)$, using Lemma~\ref{lem:tx} and process A state probabilities;
\item Process A transition probabilities $\Pi_e'$, $\Pi_s'$, $\Pi_c'$, $\Pi_e^-$, $\Pi_s^+$, $\Pi_s^-$, $\Pi_c^+$, $\Pi_c^-$.
\item Process B transition probabilities $\Pi_e$, $\Pi_s$, $\Pi_c$.
\item Process A and Process B state probabilities $\Pr(t + 1, c, s, r)$ and $\Pr(t + 1, c, s)$, using computed transition probabilities.
\end{enumerate}
We continue the computation until the total probability of absorbing states exceeds some predefined threshold $1 - \epsilon$.

For each state, we can find the time needed to transit to this state:
\[
T(t, c, s) = c T_c + s T_s + (t - c - s) T_e.
\]
Having found probability distributions of the processes states, we obtain the sought time distributions defined in Section \ref{sec:problem}:
\[
P_A^N(\tau) = \smashoperator[l]{\sum_{t, c, s\colon T(t + 1, c, s + 1) = \tau}} \sum_{r = 0}^{RL - 1} \Pr(t, c, s, r) \Pi_s^+(t, c, s);
\]
\[
P_B^N(\tau) = \smashoperator[l]{\sum_{t, c\colon T(t + 1, c, N) = \tau}} \Pr(t, c, N-1) \Pi_s(t, c, N-1).
\]

\section{Numerical Results}
\label{sec:results}

\begin{figure}[!bp]
\centering
  \includegraphics[width=0.8\linewidth]{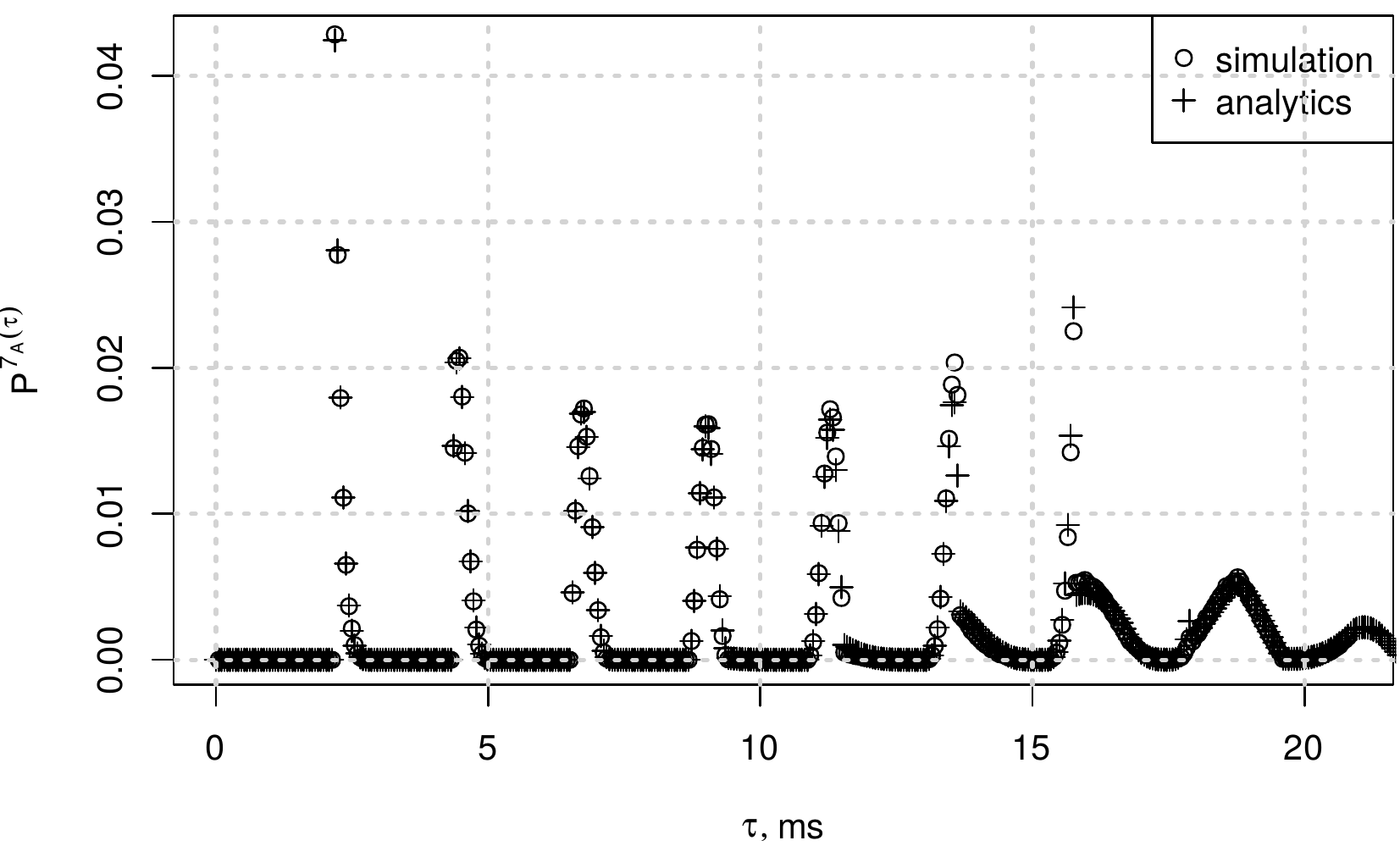}
  \includegraphics[width=0.8\linewidth]{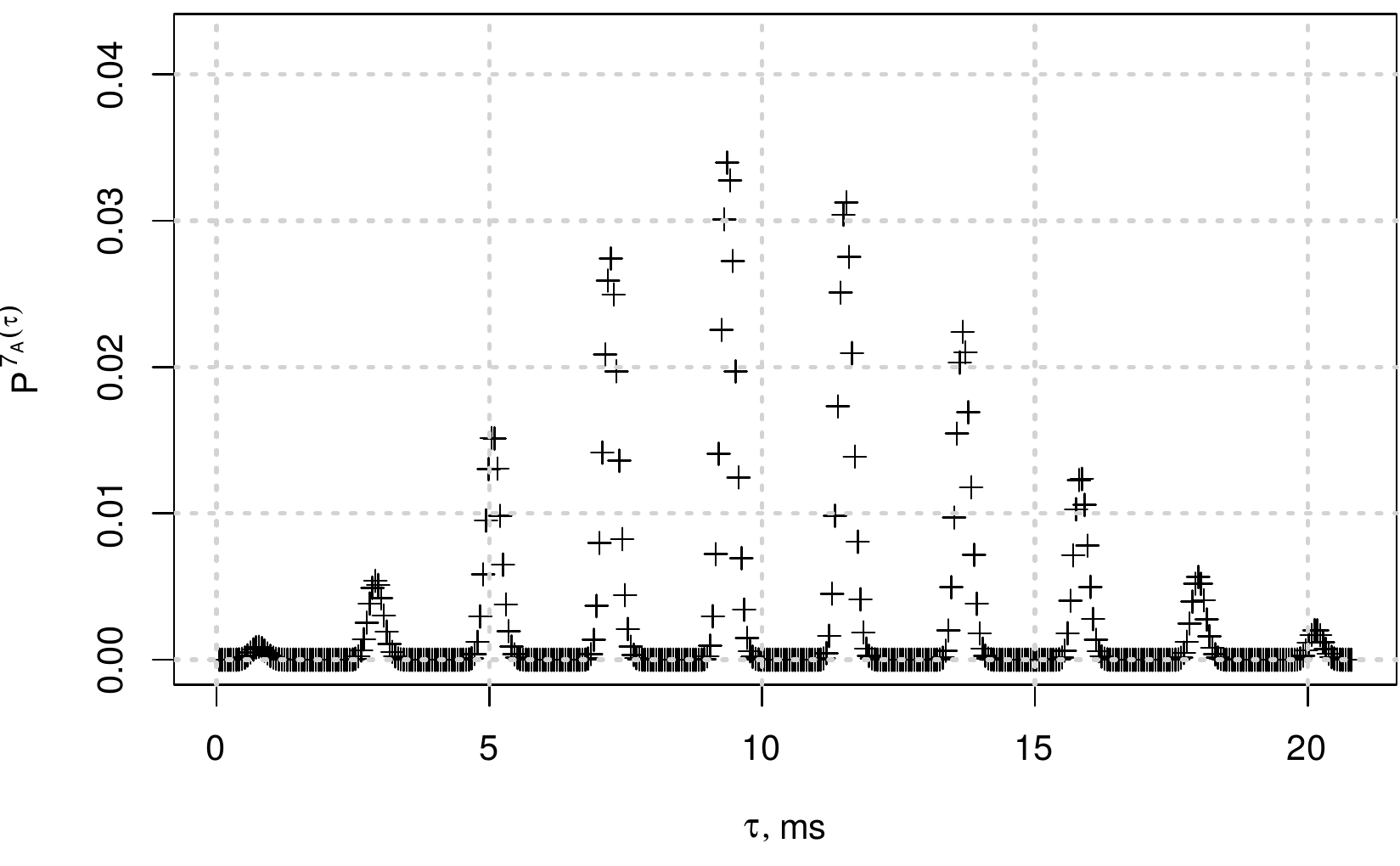}
   \caption{
   \label{fig:delay} Probability distribution  $P^7_A(\tau)$ of the time needed to one of 7 STAs to deliver its frame, obtained with the developed model and simulation (top), and with the model from \cite{liu2013power} (bottom).}
\end{figure}

To validate the model, we compare its results with the ones obtained with the well-known ns-3 simulator \cite{ns-3}. In the experiments, we consider an IEEE 802.11ah network with STAs located 1 meter from the AP, which guarantee that only collisions cause transmission errors. Each STA has a frame destined for the AP. The AP selects $N$ STAs and assigns them a RAW slot. When the RAW slot begins, the selected STAs start accessing the channel to transmit their frames. We assume infinite RAW slot duration to obtain distribution of the time needed to transmit frames.
We set $CW_{min} = 16$, $CW_{max} = 1024$, $AIFS = SIFS + 3 \sigma$, $T_e = \sigma = \SI{52}{\us}$, $RL = 7$. Data frames are 100~byte long and they are transmitted with MCS0 \cite{comcom11ah} at the 2 MHz channel. So $T_c \approx T_s \approx 42 \sigma$.

\begin{figure}[!tp]
\centering
  \includegraphics[width=0.8\linewidth]{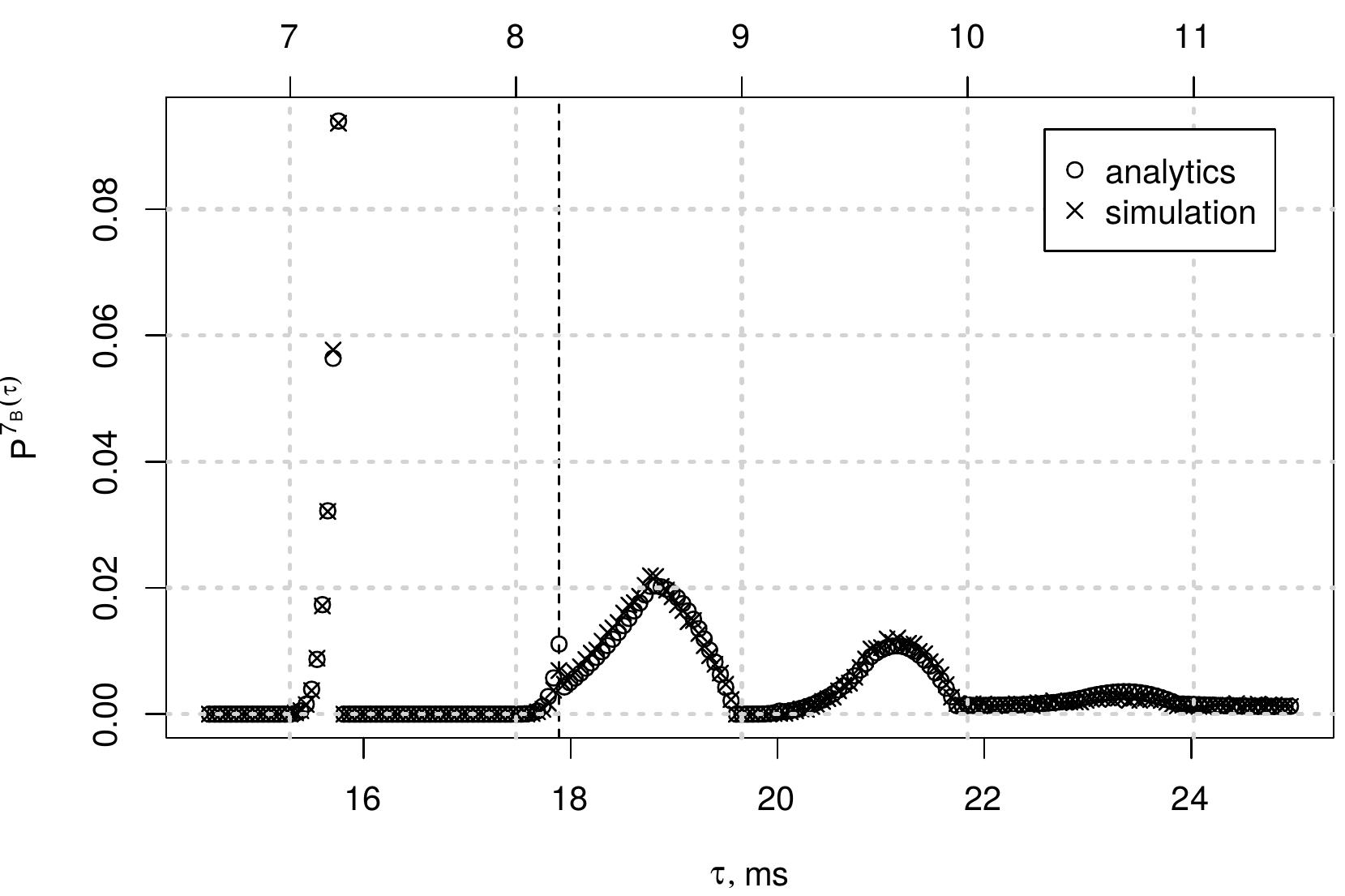}
  \includegraphics[width=0.8\linewidth]{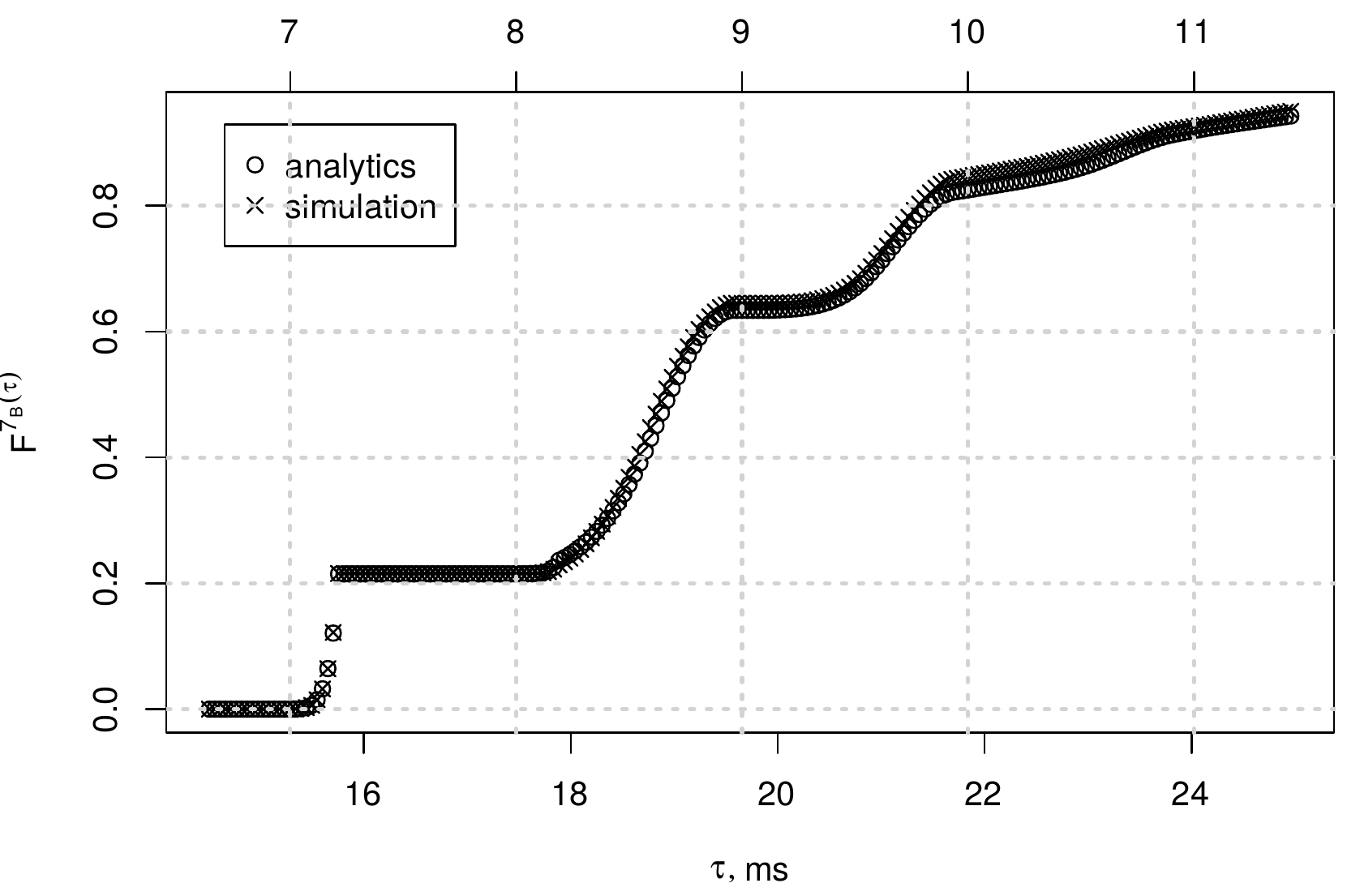}
  \caption{\label{fig:times}PDF ($P^7_B(\tau)$) and CDF ($F^7_B(\tau)$) of the time needed for 7 STAs to deliver their frames.}
\end{figure}

First of all, we find the distribution of time needed for one of $N = 7$ STAs to deliver a data frame (problem A).
The results obtained with the designed model and simulation are shown in Fig.~\ref{fig:delay} (top) and prove the high accuracy of our model.  At the same time, the model developed in \cite{liu2013power}, i.e., under the assumption that the probability of transmission in a slot depends only on the slot number, cannot be used to find the required probability distribution, see Fig.~\ref{fig:delay} (bottom). In both figures, intervals between peaks correspond to a non-empty virtual slot duration.

Fig.~\ref{fig:times} shows the probability distribution function $P_B(\tau)$ and cumulative distribution function $F_B(\tau)$ of the time needed for 7 STAs to deliver their frames.
The plots have two horizontal axes. The top  axis shows time in $T_s$, while the time unit of the bottom one is ms.

Let us analyze the obtained results in details. In particular, the first peak of $P_B(\tau)$ (between $7 T_s$ and $8 T_s$) corresponds to cases when frames are transmitted without collisions.

The time between $8 T_s$ and $9 T_s$ corresponds to eight transmission attempts, i.e. to the cases with the only collision of two frames. Two peaks in this interval correspond to cases when the last slot is occupied by the first transmission attempt (at $\approx$ 18 ms) or the second transmission attempt (at $\approx 19$ms). In the first case, the time of process termination is determined by the backoff time distribution for the first transmission attempt.  Therefore, the peak is sharp, similar to the first peak. In the second case, the peak is sawtooth-like as it is determined by the convolution of the first and the second backoff times.

According to the idea introduced in Section \ref{sec:problem}, we choose the RAW slot duration in such a way that frame transmission by a chosen STA or by all STAs ends by the end of the slot with some predefined probability. Fig.~\ref{fig:median} shows $0.5$ quantile (median), $0.95$ quantile, $0.99$ quantile and $0.999$ quantile of process A termination time for various number of STAs. Similar curves were obtained for process B. Based on these results, we can select an appropriate RAW slot duration for any number of STAs.

\begin{figure}[!htbp]
\centering
\includegraphics[width=0.8\linewidth]{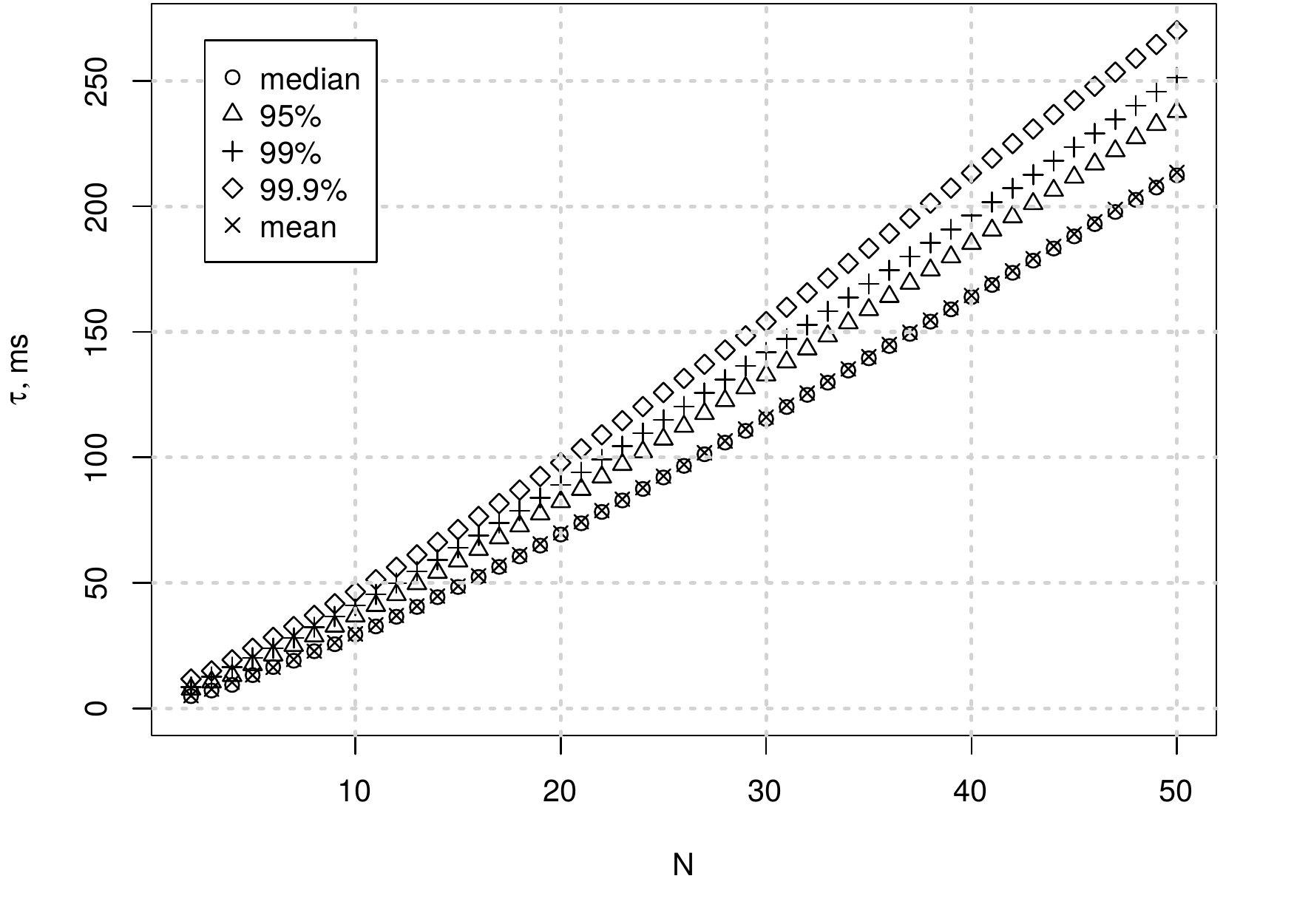}
\caption{\label{fig:median} Process A termination time.}
\end{figure}

Let us consider a more natural scenario with an AP and $N=1000$ sensor STAs, each of which has a packet for transmission with probability $p=0.3$ by the RAW slot beginning. It is obvious that PDF of the time needed to one the STAs to deliver its frame is $P_{A}^{(N,p)}(\tau)=\sum\limits_{k}P_A^k(\tau) \binom {n} {k} p^k(1-p)^{(1-k)}$.
If we need to guarantee $F_{A}^{(N,p)}(\tau)$ higher than some predefined threshold, say 0.9, we can select the required RAW slot duration accordingly, see Fig.~\ref{fig:complex}. However, the obtained RAW slot is too long and even exceeds the maximal RAW slot duration allowed by the standard (246.14 ms).

\begin{figure}[!htbp]
\centering
\includegraphics[width=0.8\linewidth]{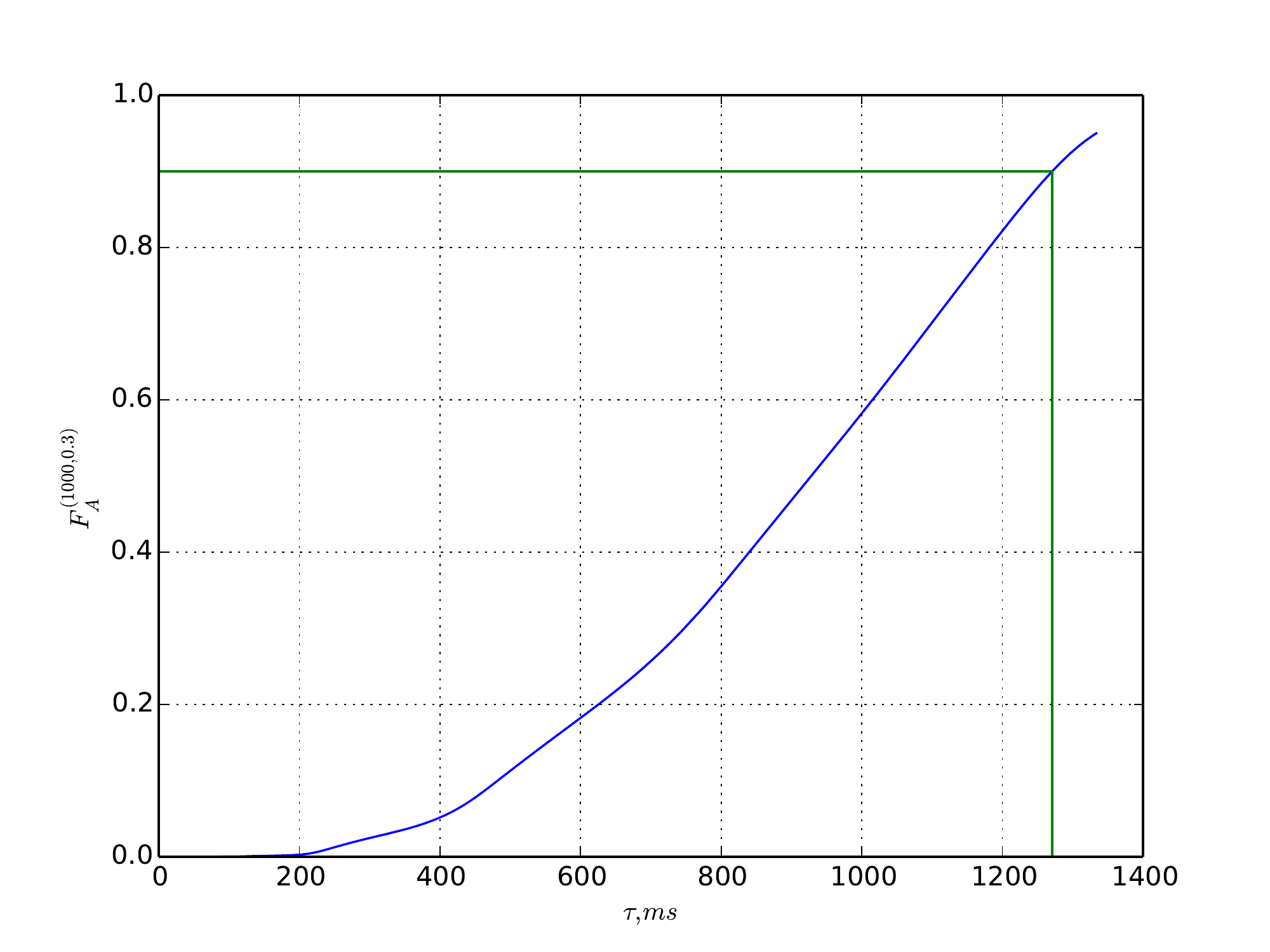}
\caption{\label{fig:complex} CDF $F_{A}^{(1000,0.3)}(\tau)$ of the time needed to a STA to transmit its packet, when each of 1000 STAs has a packet for transmission with probability 0.3.}
\end{figure}

To cope with this, let us divide $N$ STAs into $g$ equal groups and compare the required total duration of RAW slots for various $g$. Apart from moving the slot to the legal range, such a division allows to save channel resources. In particular, as  Fig.~\ref{fig:complex2} shows, the usage of the only group requires 35\% more channel time than when the STAs are divided into 40-50 groups. With higher $N$ and $p$, the effect is more significant.

\begin{figure}[!htbp]
\centering
\includegraphics[width=0.8\linewidth]{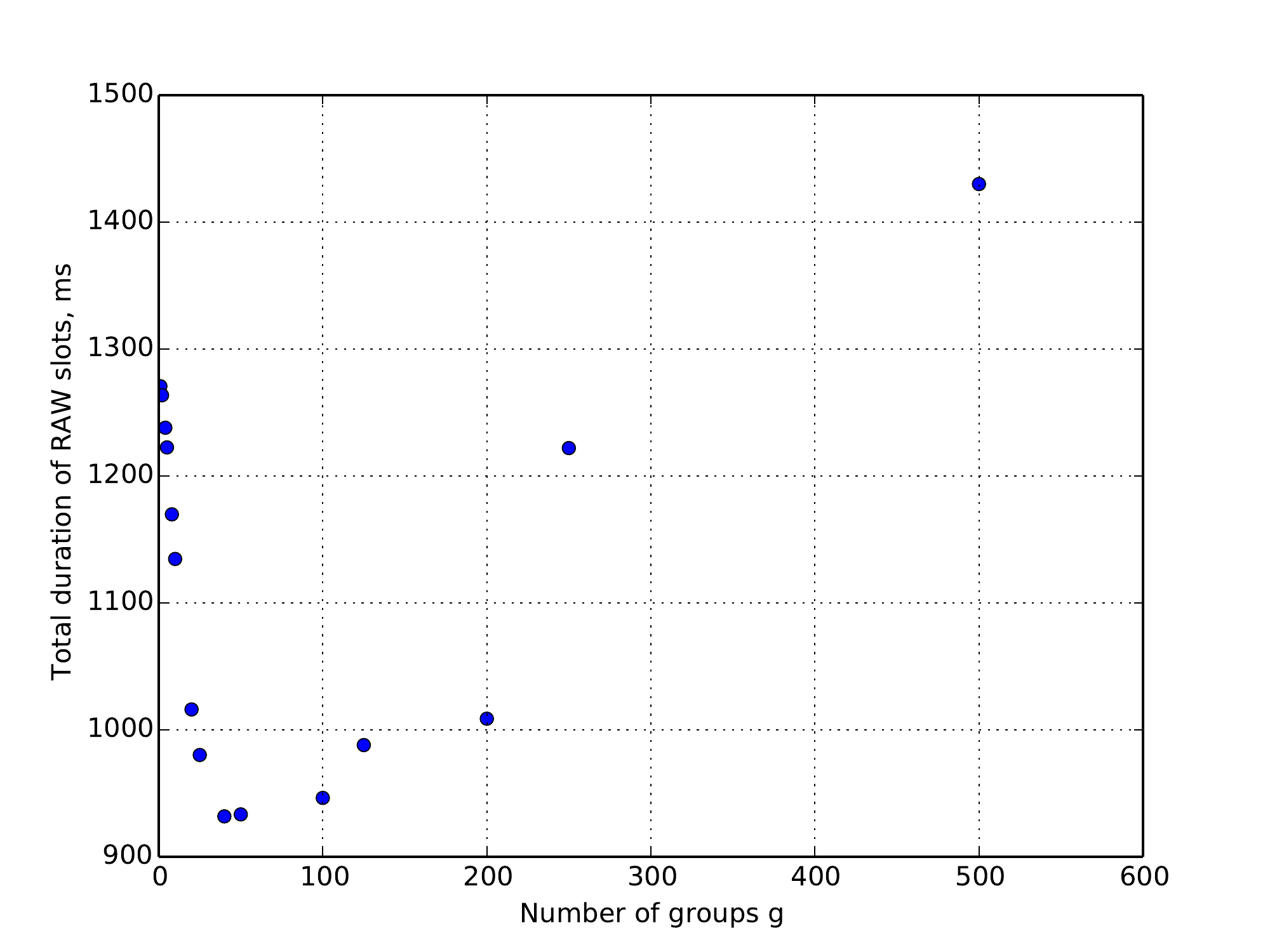}
\caption{\label{fig:complex2} Minimizing reserved channel time}
\end{figure}

\section{Conclusion}
\label{sec:concl}
In this paper, we have studied the process of the channel access by a group of STAs, each of which has the only frame to transmit. The STAs start accessing the channel simultaneously using random channel access. Analysis of such a process is very useful for performance evaluation and improvement of MTC data transmission in IEEE 802.11ah networks. In particular, a similar process takes place when a group of STAs transmit their frames in the RAW slot. Another example is the power management mechanism and sending PS-Polls to retrieve buffered frames stored at the AP.

Specifically, in the paper, we studied the access-in-RAW process from two points of view and developed a mathematical model which allows to find (A) the distribution of time required for an arbitrarily chosen STA to transmit its frame, and (B) the distribution of time required for all STAs to transmit their frames. After estimating the number of STAs having frames to transmit, the AP may use the model to select an appropriate RAW slot duration. We have also demonstrated, how to divide STAs into groups and how to select an appropriate RAW slot duration if the exact number of STAs having packets to transmit is unknown.

\addcontentsline{toc}{section}{Literature}

\bibliographystyle{unsrt}
\bibliography{raw}

\begin{thebibliography}{1}

\bibitem{802.11ah}
{\em {IEEE P802.11ahTM/D3.0 Draft Standard for Information technology —-
  Telecommunications and information exchange between systems Local and
  metropolitan area networks -- Specific requirements -- Part 11: Wireless LAN
  Medium Access Control (MAC) and Physical Layer (PHY) Specifications --
  Amendment 6: Sub 1 GHz License Exempt Operation}}, September 2014.

\bibitem{automation}
{A. S. Ivanov, E. M. Khorov, A. I. Lyakhov}.
\newblock Analytical model of batch flow multihop transmission in wireless
  networks with channel reservations.
\newblock {\em Automation and Remote Control}, July 2015, In Press.

\bibitem{comcom11ah}
{Evgeny Khorov, Andrey Lyakhov, Alexander Krotov, Andrey Guschin}.
\newblock {A survey on IEEE 802.11ah: an Enabling Networking Technology for
  Smart Cities}.
\newblock {\em Computer Communications}, 58:53--69, March 2015.

\bibitem{bianchi2000performance}
Giuseppe Bianchi.
\newblock Performance analysis of the ieee 802.11 distributed coordination
  function.
\newblock {\em Selected Areas in Communications, IEEE Journal on},
  18(3):535--547, 2000.

\bibitem{zheng2013performance}
Lei Zheng, Lin Cai, Jianping Pan, and Minming Ni.
\newblock Performance analysis of grouping strategy for dense ieee 802.11
  networks.
\newblock {\em Proc. of IEEE GLOBECOM’13}, pages 1--6, 2013.

\bibitem{buratti2010be}
Chiara Buratti.
\newblock Performance analysis of ieee 802.15. 4 beacon-enabled mode.
\newblock {\em Vehicular Technology, IEEE Transactions on}, 59(4):2031--2045,
  2010.

\bibitem{liu2013power}
Ren~Ping Liu, Gordon~J Sutton, and Iain~B Collings.
\newblock Power save with offset listen interval for {IEEE} 802.11 ah smart
  grid communications.
\newblock In {\em Communications (ICC), 2013 IEEE International Conference on},
  pages 4488--4492. IEEE, 2013.

\bibitem{ns-3}
The ns-3 network simulator, http://www.nsnam.org/.

\end{thebibliography}

\end{document}